\documentclass[a4paper]{scrartcl}

\usepackage[utf8]{inputenc}
\usepackage[T1]{fontenc}
\usepackage{lmodern}
\usepackage{subcaption}
\usepackage{tikz}

\usepackage{amsmath}
\usepackage{amssymb}
\usepackage{amsthm}

\theoremstyle{theorem}
\newtheorem{theorem}{Theorem}
\newtheorem{lemma}[theorem]{Lemma}
\newtheorem{corollary}[theorem]{Corollary}

\theoremstyle{definition}
\newtheorem{remark}[theorem]{Remark}
\newtheorem{example}[theorem]{Example}

\DeclareMathOperator{\st}{\,s.t.}

\DeclareMathOperator{\argmax}{argmax}

\usepackage[ruled,vlined,linesnumbered]{algorithm2e}
\SetKwInOut{Input}{Input}
\SetKwInOut{Output}{Output}
\DontPrintSemicolon

\usepackage{hyperref}

\begin{document}

\title{The robust bilevel continuous knapsack problem with uncertain
  coefficients in the follower's objective\thanks{This work has
    partially been supported by Deutsche Forschungsgemeinschaft (DFG)
    under grants no.~BU 2313/2 and~BU 2313/6. This version of the
    article has been accepted for publication. The Version of Record
    is available online at:
    \url{http://dx.doi.org/10.1007/s10898-021-01117-9}}}

\author{Christoph Buchheim, Dorothee Henke \\[1ex]
  Department of Mathematics, TU Dortmund University, Germany \\[1ex]
  christoph.buchheim,dorothee.henke@math.tu-dortmund.de}

\date{}

\maketitle

\begin{abstract}
  We consider a bilevel continuous knapsack problem where the
  leader controls the capacity of the knapsack and the follower
  chooses an optimal packing according to his own profits, which may
  differ from those of the leader. To this bilevel problem, we add
  uncertainty in a natural way, assuming that the leader does not have
  full knowledge about the follower's problem. More precisely,
  adopting the robust optimization approach and assuming that the
  follower's profits belong to a given uncertainty set, our aim is to
  compute a solution that optimizes the worst-case
  follower's reaction from the leader's perspective. By
  investigating the complexity of this problem with respect to
  different types of uncertainty sets, we make first steps towards
  better understanding the combination of bilevel optimization
  and robust combinatorial optimization. We show that the problem
  can be solved in polynomial time for both discrete and interval
  uncertainty, but that the same problem becomes NP-hard when each
  coefficient can independently assume only a finite number of
  values. In particular, this demonstrates that replacing uncertainty
  sets by their convex hulls may change the problem significantly, in
  contrast to the situation in classical single-level robust
  optimization. For general polytopal uncertainty, the problem again
  turns out to be NP-hard, and the same is true for ellipsoidal
  uncertainty even in the uncorrelated case. All presented hardness
  results already apply to the evaluation of the leader's objective
  function.

\end{abstract}

\section{Introduction}

Bilevel optimization has received increasing attention in the last
decades. The aim is to model situations where certain decisions are
taken by a so-called leader, but then one or more followers optimize
their own objective functions subject to the choices of the
leader. The follower's decisions in turn influence the leader's
objective, or even the feasibility of her decisions. The objective is
to determine an optimal decision from the leader's perspective. In
general, bilevel optimization problems are very hard to solve. Even in
the case that both objectives and all constraints are linear, the
bilevel problem turns out to be strongly NP-hard in
general~\cite{Hansenetal}. Several surveys and books on bilevel
optimization have been published recently,
e.g.,~\cite{Colsonetal,Dempe_bibliography,Dempeetal}.

Our research is motivated by complexity questions. We investigate
whether -- and in which cases -- taking uncertainty of some problem
parameters into account renders a bilevel optimization problem
significantly harder, where we focus on the robust optimization
approach. In this approach, the uncertain parameters are specified by
so-called uncertainty sets which contain all possible (or likely)
values the parameters can attain; the aim is to find a solution that
is feasible in each of these scenarios and that optimizes the worst
case. Typical types of uncertainty sets that are used in robust
optimization are finite sets, (general or special) polytopes, and
ellipsoids.

In our investigation of the bilevel continuous knapsack problem, we
only consider uncertainty in the coefficients of the follower's
objective function. This is a very natural setting in bilevel
optimization, as in practice the leader often does not know the
follower's objective function precisely.
A few other possible combinations of bilevel and robust optimization
should be mentioned, in contrast to ours, as it can make a big
difference what parameters are uncertain and from whose
perspective. As already mentioned, we will investigate the case of an
uncertain follower's objective function from the leader's
perspective. One could also look at uncertainty in the leader's
objective function (still from the leader's perspective) which
corresponds to the single-level robust optimization setting with
uncertain objective, ignoring the bilevel structure of the leader's
problem. Yet another option would be to make the follower's problem
uncertain in the single-level robust optimization sense, i.e., to
assume that the follower's objective function is uncertain to the
follower himself. In our opinion, the setting we chose is the most
natural one and also the most interesting one because no results from
single-level robust optimization can be applied directly.

Due to the hardness of deterministic bilevel optimization, it is not
surprising that relatively few articles dealing with bilevel
optimization problems under uncertainty, in particular using the
robust optimization approach, have been published so far.
In~\cite{ChuongJeyakumar}, the authors consider bilevel problems with
linear constraints and a linear follower's objective, while the
leader's objective is a polynomial. The robust counterpart of the
problem, with interval uncertainty in the leader's and the follower's
constraints, is solved via a sequence of semidefinite programming
relaxations. In~\cite{SariddichainuntaInuiguchi}, a bilevel linear
programming problem with the follower's objective vector lying in a
polytopal uncertainty set is considered. A vertex enumeration approach
is combined with a global optimality test based on an inner
approximation method.  Similar models have also been considered in a
game-theoretic context; see, e.g.,~\cite{HuFukushima}. More literature
has been published on the stochastic optimization approach to bilevel
optimization under uncertainty, where (some of) the problem's
parameters are assumed to be random variables and the aim is to
determine a solution optimizing the expected objective value; see,
e.g.,~\cite{Henkel} and the references therein.

Contrary to our setting of an uncertain follower's objective
function, one can assume that the leader knows the follower's
objective function precisely, but that the follower does not
necessarily choose his optimal solution. This leads to a closely
related mathematical model, which corresponds to situations where
the follower acts with so-called bounded rationality; this framework
is also known as decision variable uncertainty.  Multi\-level
optimization problems in which lower-level decision makers might
deviate from their optimal values by a small amount are investigated
in~\cite{Besancon2019,Besancon2020} under the notion of
``near-optimal robustness''. The authors show that the complexity
does not increase when adding this type of uncertainty to a
model~\cite{Besancon2020}. Similarly, in~\cite{Zare2018}, it is
assumed that the follower may choose a solution that is suboptimal
to some degree in order to worsen the leader's objective value. The
setting of~\cite{Zare2018} can be seen as a generalization of both
pessimistic bilevel optimization and single-level robust
optimization.

In classical single-level robust optimization, even in the case of
uncertainty concerning only the objective function, some classes of
uncertainty sets may lead to substantially harder robust counterparts,
e.g., finite uncertainty sets in the context of combinatorial
optimization~\cite{kouvelis}. In other cases, the robust counterparts
can be solved by an efficient reduction to the under\-lying certain
problem. This is true in particular for the case of interval
uncertainty, where each coefficient may vary independently within some
given range. In this case, each interval may be replaced by one of its
endpoints, depending on the direction of optimization, so that the
robust counterpart is not harder than the underlying certain
problem. For an overview of complexity results in robust combinatorial
optimization under objective uncertainty, we refer the reader to the
recent survey~\cite{survey} and the references therein.
We will show that the situation in case of interval uncertainty is
more complicated in our robust bilevel optimization setting.

We concentrate on a bilevel continuous knapsack problem where the
leader only controls the capacity, while the follower chooses the
items, optimizing his own objective function. Without uncertainty,
this problem is easy to solve, which makes it a good starting point
for addressing the question of how much harder bilevel optimization
problems can become when adding uncertainty.  As mentioned above,
linear bilevel optimization problems are NP-hard in general. As a
general approach, not using any problem-specific information, bilevel
problems are usually reformulated as single-level problems before
being solved. For example, if the follower's problem is a linear
program, it can be replaced by its Karush-Kuhn-Tucker conditions,
allowing for a branch-and-bound or branch-and-cut approach that
branches over the resulting complementarity constraints; see, e.g.,
\cite{FortunyAmatMcCarl,Kleinertetal}. However, the bilevel continuous
knapsack problem can be solved in polynomial time by a direct
combinatorial algorithm; see Section~\ref{section_withoutuncertainty}.

We only deal with the continuous problem version, i.e., the follower
solving a continuous knapsack problem, which is also discussed
in~\cite{Dempeetal} as an introductory example. The variant of this
problem where the follower solves a binary knapsack problem instead is
introduced in~\cite{DempeRichter} and further investigated
in~\cite{BrotcorneHanafiMansi}. Both articles devise pseudopolynomial
algorithms for the problem. Other variants of bilevel knapsack
problems include one where leader and follower each control a subset
of the items and pack them into a common knapsack of fixed
capacity~\cite{Mansietal}, and one where leader and follower each have
their own knapsack, but choose items from a common item
set~\cite{DeNegre}. All three variants are also investigated
in~\cite{Capraraetal} in terms of complexity and approximability.
Algorithms for the continuous version of the problem in~\cite{DeNegre}
can be found in~\cite{CarvalhoLodiMarcotte,FischerWoeginger}.  A
stochastic version of the bilevel knapsack problem
of~\cite{DempeRichter} is considered in~\cite{Ozaltin}: it is assumed
that the capacity in the follower's problem depends on the leader's
decision as well as an uncertain value, which is drawn from a discrete
probability distribution.

A practical motivation to investigate the bilevel continuous knapsack
problem under uncertainty is the possible application in revenue
management that is also described in~\cite{BrotcorneHanafiMansi}
and~\cite{Mansietal} for the respective variants of bilevel knapsack
problems: the leader, who might be an investor or a company, hands a
part of her savings over to the follower, an intermediary, who invests
it, maximizing his own profit and giving back part of the returns to
the leader. It is very natural to consider continuous variables in
this application. Moreover, it is reasonable to consider the
follower's objective uncertain, due to the inherent risk of
investments and hence the investor's uncertainty about the
intermediary's assessment of the expected returns.

While the certain bilevel continuous knapsack problem is easily
solvable, it becomes much more involved when adding uncertainty in
the follower's objective. It turns out that the interval
uncertainty setting requires to deal with partial orders, more
precisely, with the interval orders induced by the relations between
the follower's profit ranges. Adapting an algorithm by
Woeginger~\cite{Woeginger} for some precedence constraint knapsack
problem, we show that the problem can still be solved in polynomial
time. We also discuss why the case of finite uncertainty sets is
tractable as well.  For many other types of uncertainty sets, the
robust bilevel continuous knapsack problem turns out to be NP-hard,
and the same is true even for the problem of evaluating the leader's
objective function, i.e., the adversary's optimization problem.
We are able to show this for all common types of uncertainty sets
and even special cases of them: we consider simplices as uncertainty
sets, thus settling the more general cases of polytopal and
Gamma-uncertainty, and uncertainty sets defined by norms, which
includes the case of ellipsoidal uncertainty. For an overview and
formal definitions of these classes of uncertainty sets, see,
e.g.,~\cite{survey}.

Our results also emphasize another significant difference to classical
robust optimization: in general, it is not possible anymore to replace
the uncertainty set by its convex hull and thus assume convexity
without loss of generality. In fact, when restricting the possible
scenarios in the interval case to the endpoints of the intervals, we
show that the problem turns NP-hard. More precisely, the problem is
NP-hard when the input consists of a finite set of realizations for
each coefficient and these realizations arise independently.

\enlargethispage{\baselineskip}

The remainder of this paper is organized as follows. In
Section~\ref{section_withoutuncertainty}, we introduce and discuss the
certain variant of the problem. We then settle the cases of finite
uncertainty sets in Section~\ref{section_scenariouncertainty} and
interval uncertainty in Section~\ref{section_intervaluncertainty}. The
discrete uncorrelated case, where each coefficient varies in a finite
set independently, is discussed in
Section~\ref{section_discreteuncorrelateduncertainty}. Finally,
uncertainty sets defined as simplices
(Section~\ref{section_simplicialuncertainty}) or by norms
(Section~\ref{section_normuncertainty}) are
discussed. Section~\ref{section_conclusion} concludes.

\section{Underlying certain problem}\label{section_withoutuncertainty}

\subsection{Problem formulation}

We first discuss the deterministic variant of the bilevel optimization
problem under consideration, in which the follower solves a continuous
knapsack problem, while the leader determines the knapsack's capacity
and optimizes another linear objective function than the
follower. This problem is also discussed in~\cite{Dempeetal}, but we
replicate the formulation and the algorithm here for sake of
completeness.

First recall that an important issue in bilevel optimization is that
the follower's optimal solution is not necessarily unique, but the
choice among the optimal solutions might have an impact on the
leader. The two main approaches here are the optimistic and the
pessimistic one. In the former case, the follower is assumed to decide
in favor of the leader, whereas in the latter case, he chooses the
optimal solution that is worst for the leader. For more details, see
e.g.,~\cite{Dempeetal,Wiesemannetal}.  While often only the optimistic
approach is considered, we will focus on the pessimistic one in this
paper, since it combines more naturally with the concept of robustness
that will be added to the problem later on. Indeed, both the
pessimistic view of the follower and the robustness force the leader
to consider the worst case with respect to some set of
choices. However, all our results hold for the optimistic approach as
well, making only small changes to the proofs necessary, which will be
sketched for every result.

The pessimistic version of the problem, without uncertainty, can be
formulated as follows, where the minimization over $x$ represents the
leader's pessimism about the follower's choice among his optimal
solutions:
\begin{equation}\tag{P}\label{certainprob}
\begin{aligned}
  & \max_{b\in [b^-,b^+]} & & \min_{x} & & d^\top x - \delta b\\
  & & & \st & & x \in
  \begin{aligned}[t]
    & \argmax & & c^\top x\\
    & \st & & a^\top x \leq b \\
    & & & 0 \leq x \leq 1 \\
  \end{aligned} \\
\end{aligned}
\end{equation}
The leader's only variable is $b \in \mathbb{R}$, which can be
considered as the knapsack's capacity. The follower's variables are $x
\in \mathbb{R}^n$, i.e., the follower fills the knapsack with a subset
of the objects, where also fractions are allowed. The item sizes $a
\in \mathbb{R}_{\geq 0}^n$, the follower's item values~$c \in
\mathbb{R}_{> 0}^n$, the capacity bounds $b^-, b^+ \in \mathbb{R}$ as
well as the leader's item values $d \in \mathbb{R}_{\geq 0}^n$ and a
number~$\delta \geq 0$ are given. The latter can be thought of as a
price the leader has to pay for providing one unit of knapsack
capacity.  We may assume $0 \leq b^- \leq b^+ \leq \sum_{i = 1}^n a_i$
and~$a > 0$.

For a given leader's choice $b$, due to the assumptions~$0 \leq b \leq
\sum_{i = 1}^n a_i$ and $c > 0$, every optimal solution of the
follower's problem satisfies $a^\top x = b$. This can be used to show
that the following three ways to formulate the leader's objective
function are equivalent:
\begin{itemize}
\item[(a)] $d \in \mathbb{R}^n$ and $\delta \in \mathbb{R}$, the most
  general variant,
\item[(b)] $d \in \mathbb{R}_{\geq 0}^n$ and $\delta \geq 0$, as in the
  problem formulation above, and
\item[(c)] $d \in \mathbb{R}^n$ and $\delta = 0$, which we will use
  in the following.
\end{itemize}
Clearly, (b) and (c) are special cases of (a). An objective function $d^\top
x - \delta b$ of type~(a) can be reformulated in the form of (c):
$$d^\top x - \delta b = d^\top x - \delta (a^\top x) = (d - \delta a)^\top x$$
and similarly in the form of (b):
$$d^\top x - \delta b = (d - \delta a)^\top x = (d - \varepsilon a - (\delta - \varepsilon) a)^\top x = (d - \varepsilon a)^\top x - (\delta - \varepsilon) b,$$
where $\varepsilon := \min\{\delta, \min_{i \in \{1, \dots,
  n\}}\frac{d_i}{a_i}\}$.  This proves that all three formulations are
equivalent.  Note that this equivalence still holds for the uncertain
problem versions considered later on.  As already indicated, we will
use the third formulation from now on because it is the most compact
one, i.e., we omit $\delta$ and assume $d \in \mathbb{R}^n$.

\subsection{Solution algorithm}

The follower solves a continuous knapsack problem with fixed capacity
$b$. This can be done, for example, using Dantzig's
algorithm~\cite{Dantzig}: by first sorting the items, we may assume
$$\tfrac{c_1}{a_1} \geq \dots \geq \tfrac{c_n}{a_n}\;.$$
The idea is then to
pack the items into the knapsack in this order until it is full. More
formally, if $b = \sum_{i = 1}^n a_i$, everything can be taken, so the
optimal solution is $x_i = 1$ for all~$i \in \{1, \dots,
n\}$. Otherwise, we consider the critical item
\[\textstyle k := \min\big\{j \in \{1, \dots, n\} \colon \sum_{i = 1}^j a_i > b\big\}\]
and an optimal solution is given by
\[x_j:=\begin{cases}\begin{array}{ll}
  1 & \text{ for }j \in \{1, \dots, k - 1\}\\
  \tfrac{1}{a_k}\big(b - \sum_{i = 1}^{k - 1} a_i\big) & \text{ for }j=k\\
  0 & \text{ for }j \in \{k+1, \dots, n\}\;.
\end{array}\end{cases}\]

We now turn to the leader's perspective. As only the critical
item~$k$, but not the sorting depends on $b$, the leader can just
compute the described order of items by the values $\frac{c_i}{a_i}$
and then consider the behavior of the follower's optimal solution~$x$
when $b$ changes. The leader's objective function $f$ is given by the
corresponding values $d^\top x$:
$$f(b) := \sum_{i = 1}^{j - 1} d_i + \frac{d_j}{a_j}\Big(b - \sum_{i = 1}^{j - 1} a_i\Big) \text{ for }b\in\Big[\sum_{i = 1}^{j - 1} a_i,\sum_{i = 1}^{j} a_i\Big],~j\in\{1,\dots,n\}$$
Note that this piecewise linear function is well-defined and
continuous with vertices in the points $b = \sum_{i = 1}^j a_i$, in
which the critical item changes from $j$ to $j + 1$.  The leader has
to maximize $f$ over the range $[b^-, b^+]$. As $f$ is piecewise
linear, it suffices to evaluate it at the boundary points~$b^-$
and~$b^+$ and at the feasible vertices~$b = \sum_{i = 1}^j
a_i$ for all $j \in \{0, \dots, n\}$ with~$\sum_{i = 1}^j
a_i\in[b^-,b^+]$. Hence, Problem~\eqref{certainprob} can be solved in
time $\mathcal{O}(n \log n)$, which is the time needed for sorting.

\begin{remark}
  The order of items and hence the follower's optimal
  solution is not unique if the profits $c_i/a_i$ are not all
  different. In the optimistic approach, a follower would sort the
  items with the same profit in descending order of the values
  $d_i/a_i$, in the pessimistic setting in ascending order. If this is
  still not unique, there is no difference for the leader either.
\end{remark}

\section{Discrete uncertainty} \label{section_scenariouncertainty}

Turning to the uncertain problem variant, we first consider the robust
version of the problem where the follower's objective function is
uncertain for the leader, and this uncertainty is given by a finite
uncertainty set~$U \subset \mathbb{R}_{> 0}^n$ containing the possible
objective vectors~$c$.  We obtain the following problem formulation:
\begin{equation*}
\begin{aligned}
  & \max_{b\in [b^-,b^+]} & & \min_{c, x} & & d^\top x \\[-1ex]
  & & & \st & & c \in U \\
  & & & & & x \in
  \begin{aligned}[t]
    & \argmax & & c^\top x\\
    & \st & & a^\top x \leq b \\
    & & & 0 \leq x \leq 1 \\
  \end{aligned} \\
\end{aligned}
\end{equation*}
The inner minimization problem can be interpreted as being controlled
by an adversary, thus leading to an optimization problem involving
three actors: first, the leader takes her decision~$b$, then the
adversary chooses a follower's objective~$c$ leading to a follower's
solution that is worst possible for the leader, and finally the
follower optimizes this objective choosing~$x$. In the pessimistic
view of the bilevel problem, the adversary can be assumed to also
choose among the follower's optimal solutions if this is not unique,
i.e., to minimize over $c$ as well as $x$.

Note that the robust uncertainty we consider really concerns the
interplay of leader and follower in the bilevel problem structure and
that the setting differs greatly from just combining a bilevel
optimization problem with a robust optimization problem. In
particular, the adversary's problem is not the same as in a continuous
knapsack problem under uncertainty because he is the leader's and not
the follower's adversary, having the inverse of the leader's objective
as his objective function.

Again, we aim at solving the problem from the leader's perspective,
which can be done as follows: for every $c \in U$, consider the
piecewise linear function~$f_c$ as described in
Section~\ref{section_withoutuncertainty}. The vertices of each~$f_c$
can be computed in~$\mathcal{O}(n \log n)$ time, both in the
optimistic and the pessimistic approach. The leader's objective
function is then the pointwise minimum~$f:=\min_{c\in U}f_c$ and her
task is to maximize $f$ over~$[b^-,b^+]$.
The function $f$ can be seen as the lower envelope
of $|U| n$ many linear segments, corresponding to the pieces of the
functions $f_c$, and can thus be computed in time $\mathcal{O}(|U|n
\log(|U|n))$~\cite{Hershberger}. This proves

\begin{theorem}\label{theorem_finite}
  The robust bilevel continuous knapsack problem with finite
  uncertainty set~$U$ can be solved in $\mathcal{O}(|U|n \log(|U|n))$
  time.
\end{theorem}

Note that the tractability of the robust bilevel continuous knapsack
problem with finite uncertainty set~$U$ can also be obtained in a
straightforward way, at the cost of a higher running time: in order to
maximize the piecewise linear function~$f$ over~$[b^-,b^+]$, it
suffices to consider the endpoints~$b^-,b^+$, the vertices of the
function~$f_c$ for each~$c\in U$, and the intersection between each
pair of linear pieces of the functions~$f_c$ and~$f_{c'}$
with~$c,c'\in U$. Obviously, all these candidates can be examined in
polynomial time.

\section{Interval uncertainty}\label{section_intervaluncertainty}

We next address a robust version of the problem having the same
structure as in Section~\ref{section_scenariouncertainty}, but now the
uncertainty is given by an interval for each component of~$c$. We thus
consider~$U = [c_1^-, c_1^+] \times \dots \times [c_n^-, c_n^+]$ and
assume $0 < c^- \leq c^+$.  In classical robust optimization, one
could just replace each uncertain coefficient~$c_i$ by an appropriate
endpoint $c_i^-$ or $c_i^+$ and obtain a certain problem
again. However, such a replacement is not a valid reformulation in the
bilevel context. We will show that, in fact, the situation in the
bilevel case is more complicated, even though we can still devise an
efficient algorithm.

To simplify the notation, we define
$p^-_i:=-c_i^+/a_i$ and $p^+_i:=-c_i^-/a_i$ for
the remainder of this section. It turns out that interval orders
defined by the intervals~$[p^-_i,p^+_i]$ play a crucial role in the
investigation.

\subsection{Interval orders and precedence constraint knapsack problems}

For the leader, the exact entries of~$c_i$ in their
intervals~$[c^-_i,c^+_i]$ do not matter, but only the induced sorting
that the follower will use. The follower sorts the items by their
values $c_i/a_i$ and we therefore have to consider the
intervals~$[c_i^-/a_i, c_i^+/a_i]$ induced by the uncertainty
set. Intuitively speaking, two situations can arise for a pair $(i,
j)$ of items: either, their corresponding intervals are disjoint, say,
$c_i^-/a_i > c_j^+/a_j$. In this case, $i$ will precede $j$ in every
sorting induced by some adversary's choice $c \in U$. Otherwise, the
two intervals intersect. Then, the adversary can decide which of the
two items comes first by choosing the values $c_i$ and $c_j$
appropriately.

More formally, given $U$ and $a$, the possible sortings are exactly
the linear extensions of the partial order~$P$ that is induced by the
intervals~$[p_i^-, p_i^+]$, in the sense that we set
\[i <_P j \quad:\Leftrightarrow \quad \tfrac{c_i^-}{a_i} >
\tfrac{c_j^+}{a_j} \quad \Leftrightarrow \quad p_i^+ < p_j^-\;.\] Such
a partial order is called an \emph{interval order}.  Note that the
values $c_i/a_i$ are actually sorted in decreasing order by the
follower, but it is more common to read an interval order from left to
right. Therefore, by the definition of $p_i^-$ and $p_i^+$, the
intervals were flipped, so that we can think of the follower sorting
the negative values $-c_i/a_i$ increasingly.

Note that, in the pessimistic problem version, intervals intersecting
in only one point do not have to be treated differently from intervals
intersecting properly, since the adversary choosing the same value
$c_i / a_i$ for several items will result in the order of these items
that is worst for the leader anyway, due to the pessimism.

To solve the robust bilevel continuous knapsack problem with interval
uncertainty, one could compute the partial order $P$ and enumerate all
linear extensions of $P$. Every linear extension corresponds to a
sorting of the items the follower will use when the adversary has
chosen $c$ appropriately. Every sorting corresponds to a piecewise
linear function, and the leader's objective function is the pointwise
minimum of all these, as in Section~\ref{section_scenariouncertainty}.
This approach does not have polynomial runtime in general, as there
could be exponentially many linear extensions. However, it turns out
that it is not necessary to consider all linear extensions explicitly
and that the problem can still be solved in polynomial time.  We will
see that the adversary's problem for fixed~$b\in [b^-,b^+]$ is closely
related to the \emph{precedence constraint knapsack problem} or
\emph{partially ordered knapsack problem}. This is a $0$-$1$ knapsack
problem, where additionally, a partial order on the items is given and
it is only allowed to pack an item into the knapsack if all its
predecessors are also selected; see, e.g., Section~13.2
in~\cite{Kellereretal}.

For the special case where the partial order is an
interval order, Woeginger described a pseudopolynomial algorithm, see
Lemma~11 in~\cite{Woeginger}. There the problem is formulated in a
scheduling context and is called \emph{good initial set}. The
algorithm is based on the idea that every initial set (i.e., prefix of
a linear extension of the interval order) consists of
\begin{itemize}
\item a \emph{head}, which is the item whose interval has the
  rightmost left endpoint among the set,
\item all predecessors of the head in the interval order, and
\item some subset of the items whose intervals contain the left
  endpoint of the head in their interior,
\end{itemize}
assuming that all interval endpoints are pairwise distinct.

The algorithm iterates over all items as possible heads, and looks
for the optimal subset of the items whose intervals contain the
left endpoint of the head in their interior that results in an initial
set satisfying the capacity constraint. Since these items are
incomparable to each other in the interval order, each subproblem is
equivalent to an ordinary $0$-$1$ knapsack problem and can be solved
in pseudopolynomial time using dynamic programming; see,
e.g.,~\cite{Kellereretal}.  Our algorithm for the adversary's problem
is a variant of this algorithm for the continuous knapsack and uses
Dantzig's algorithm as a subroutine, therefore we will obtain
polynomial runtime.

For this, we introduce the notion of a \emph{fractional prefix} of a
partial order~$P$, which is a triple $(J, j, \lambda)$ such that $J
\subseteq \{1, \dots, n\}$, $j \in J$, $0 \leq \lambda < 1$
, and there
is an order of the items in $J$, ending with $j$, that is a prefix
of a linear extension of~$P$.  Every optimal solution of the follower,
given some~$b$ and~$c$, corresponds to a fractional prefix. The
follower's solution corresponding to a fractional prefix~$F=(J, j,
\lambda)$ is defined by
\[x^F_i:=\begin{cases}\begin{array}{ll}
  1 & \text{ for }i \in J \setminus \{j\}\\
  0 & \text{ for }i \in \{1, \dots, n\} \setminus J\\
  \lambda & \text{ for }i=j\;.
\end{array}\end{cases}\]
Additionally, there is the fractional prefix $\bar{F}$
  corresponding to the case in which all items are chosen, i.e.,
  $x^{\bar{F}}_i = 1$ for all $i \in \{1, \dots, n\}$.

Let $\mathcal{F}_P$ be the set of all fractional prefixes of the
interval order~$P$ given by~$U$ and~$a$, corresponding to the set of
all optimal follower's solutions for some~$b \in [0, \sum_{i = 1}^n
a_i]$ and some $c \in U$. Then the adversary's task is, for fixed $b$,
to choose a fractional prefix $F \in \mathcal{F}_P$ satisfying $a^\top
x^F = b$, which in the original formulation, he does implicitly by
choosing~$c \in U$ and anticipating the follower's optimal
solution under this objective. Therefore the leader's problem can be
reformulated as follows:
\[
\max_{b\in[b^-,b^+]} ~ \min_{\substack{F \in \mathcal{F}_P\\a^\top x^F
    = b}} ~~ d^\top x^F
\]
In the next subsections, we first describe an algorithm to solve the
inner minimization problem for fixed~$b$, i.e., the adversary's
problem, which will then be generalized to the maximization problem
over $b$.

\subsection{Solving the adversary's problem} \label{subsection_intervaladversary}

First, consider the special case where the interval order has no
relations. This means that all intervals intersect and hence, all
permutations are valid linear extensions. Note that a pairwise
intersection of intervals implies that all intervals have a common
intersection, since for every two intervals $[z_1^-, z_1^+]$ and
$[z_2^-, z_2^+]$ holds that $z_i^- \leq z_j^+$ for $i, j \in \{1,
2\}$, so the smallest right endpoint is right of or equal to the
largest left endpoint.

Then the adversary's problem is very similar to the ordinary
continuous knapsack problem. The only differences are that the
objective is now to minimize (and not maximize), that the objective
vector $d$ may contain positive and negative entries, and that the
constraint~$a^\top x \leq b$ is replaced by $a^\top x = b$. But with
this changed, the problem can still be solved using Dantzig's
algorithm as described in Section~\ref{section_withoutuncertainty};
note that the algorithm fills the knapsack completely anyway
assuming~$b \leq \sum_{i = 1}^n a_i$ and $c > 0$. Denote this
algorithm, returning the corresponding fractional prefix, by
\textsc{Dantzig}. We will also need this algorithm as a subroutine on
a subset of the item set (like the pseudopolynomial knapsack algorithm
in Woeginger's algorithm). Therefore, we consider \textsc{Dantzig} as
having input~$I \subseteq \{1, \dots, n\}$, $d \in \mathbb{R}^{n}$, $a
\in \mathbb{R}_{> 0}^{n}$, and $b\in [0,\sum_{i \in I} a_i]$, and only
choosing a solution among the items in $I$.

We have seen that in case the interval order has no relations the
adversary's problem can be solved by $\textsc{Dantzig}(\{1, \dots,
n\}, d, a, b)$.  The general adversary's problem can now be solved by
Algorithm~\ref{algorithm_adversary}.  In the notation of Woeginger's
algorithm, the $k$-th item is the head in iteration $k$, $I_k^-$ is
the set of its predecessors, and $I_k^0$ corresponds to the intervals
containing the left endpoint of the head -- not necessarily in their
interior here, so that, in particular, also~$k \in I_k^0$.

\begin{algorithm}[htb]
  \caption{Algorithm for the adversary's problem}
  \label{algorithm_adversary}

  \Input{~$a \in \mathbb{R}_{> 0}^n$, $0 \leq b \leq \sum_{i = 1}^n
    a_i$, $d \in \mathbb{R}^n$, $p^-, p^+ \in \mathbb{R}_{< 0}^n$ with
    $p^- \leq p^+$,\\ inducing an interval order $P$}

  \Output{~$F \in \mathcal{F}_P$ with $a^\top x^F = b$ minimizing $d^\top
    x^F$}
  
  \If{$b = \sum_{i = 1}^n a_i$}{
    \Return $\bar{F}$} \label{line_trivialcase}

  $K:=\emptyset$
  
  \For{$k = 1, \dots, n$}{
    $I_k^- := \{i \in \{1, \dots, n\} \colon p_i^+ < p^-_k\}$

    $I_k^0 := \{i \in \{1, \dots, n\} \colon p_i^- \leq p_k^- \leq p_i^+\}$

    \If{$0 \leq b - \sum_{i \in I_k^-} a_i < \sum_{i \in I_k^0}
      a_i$}{ \label{line_rangecondition}

      $(J_k', j_k, \lambda_k) := \textsc{Dantzig}(I_k^0, d, a, b -
      \sum_{i \in I_k^-} a_i)$ \label{line_DantzigIk0}

      $J_k := J_k' \cup I_k^-$

      $K:=K\cup \{k\}$
    }
  }

  \Return $(J_k, j_k, \lambda_k)$ with $k = \text{argmin}\{d^\top
  x^{(J_k, j_k, \lambda_k)} \colon k \in K\}$
\end{algorithm}

The basic difference to Woeginger's algorithm is that here it is
important to have a dedicated last item of the prefix, which will be
the one possibly taken fractionally by the follower. Apart from that,
the order of the items in the prefix is not relevant. In our
construction, any element of $I_k^0$ could be this last item, in
particular it could be $k$, but it does not have to. Note that, in
Algorithm~\ref{algorithm_adversary}, the prefix constructed in
iteration~$k$ does not necessarily contain the $k$-th item, but still,
all prefixes that do contain it as their head are covered by this
iteration.

\begin{lemma}
Algorithm~\ref{algorithm_adversary} is correct.
\end{lemma}

\begin{proof}
  For $b = \sum_{i = 1}^n a_i$, the only feasible and therefore
  optimal solution is $\bar{F}$, so that the result is correct if
  the algorithm terminates in line~\ref{line_trivialcase}.

  So assume~$b < \sum_{i = 1}^n a_i$ now. The first part of the proof shows that the
  algorithm returns a feasible solution if~$K\neq\emptyset$. The
  second part proves the optimality of the returned solution and also
  that $K\neq\emptyset$ always holds.

  In each iteration $k$, $I_k^-$ is the set of predecessors of $k$ in
  the interval order $P$. The set $I_k^0$ consists of items that are
  incomparable to $k$ and to each other in $P$, since the
  corresponding intervals all contain the point $p^-_k$ by
  definition. Hence it is valid (with respect to $P$) to call
  Dantzig's algorithm in line~\ref{line_DantzigIk0} on $I_k^0$. The
  condition in line~\ref{line_rangecondition} makes sure that we only
  call the subroutine if the available capacity is in the correct
  range, i.e., if it is possible to fill the knapsack with the items
  in~$I_k^-$ and a subset of the items in~$I_k^0$ of which one item is
  taken fractionally.

  Then $(J_k, j_k, \lambda_k)$ is a fractional prefix of $P$, as all
  predecessors of $k$ and therefore also all predecessors of all $i
  \in J_k \subseteq I_k^- \cup I_k^0$ belong to $J_k$, since $p_i^-
  \leq p_k^-$ holds for them. The item $j_k$ is a valid last item of a
  prefix consisting of the items in $J_k$ because $j_k \in I_k^0$ by
  construction and therefore, there are no successors of $j_k$ in
  $J_k$. Moreover,
  $$\textstyle a^\top x^{(J_k, j_k, \lambda_k)} = \sum_{i \in
    I_k^-} a_i + (b - \sum_{i \in I_k^-} a_i) = b$$
  by construction and
  by the correctness of \textsc{Dantzig}. Therefore, for all $k \in
  K$, $(J_k, j_k, \lambda_k)$ is a feasible solution.

  Now we prove the optimality of the returned solution. Let $(J, j,
  \lambda)$ be an optimal solution (if $\bar{F}$ is optimal, then $b$
  must be $\sum_{i = 1}^n a_i$, and this case is trivial). Choose $k
  \in J$ with maximal $p^-_k$ (i.e., a head of the prefix). Then
  $I_k^- \subset J$ since $J$ is a prefix and $k \in J$, so all
  predecessors of~$k$ must be in $J$, as well. Moreover, $j \in J
  \setminus I_k^-$ as all items in $I_k^-$ have at least one successor
  (namely $k$) in $J$. By the choice of $k$, we have $J \setminus
  I_k^- \subseteq I_k^0$ and $(J \setminus I_k^-, j, \lambda)$ is a
  feasible solution of the subproblem solved by the call of
  \textsc{Dantzig} in line~\ref{line_DantzigIk0} since
  $$\textstyle a^\top x^{(J
    \setminus I_k^-, j, \lambda)} = a^\top x^{(J, j, \lambda)} -
  \sum_{i \in I_k^-} a_i = b - \sum_{i \in I_k^-} a_i\;.$$
  Thus
  \begin{align*}
    d^\top x^{(J, j, \lambda)}
    & = \textstyle\sum_{i \in I_k^-} d_i + d^\top x^{(J\setminus I_k^-, j, \lambda)}\\
    & \geq \textstyle\sum_{i \in I_k^-} d_i + d^\top x^{(J_k', j_k, \lambda_k)}
     ~=~  d^\top x^{(J_k, j_k, \lambda_k)}\;,
  \end{align*}
  which is at least the cost of any returned solution.  The second part
  of the proof also shows that~$K\neq\emptyset$. Thus, the algorithm
  always returns an optimal solution.
\end{proof}

An optimal solution of the adversary's problem in the original
formulation, i.e., a vector~$c \in U$, can be derived from the
fractional prefix $(J_k, j_k, \lambda_k)$ returned by the algorithm in
the following way:
\[c_i:=\begin{cases}\begin{array}{ll}
  c_i^+ & \text{ for }i \in J_k \setminus \{j_k\}\\
  c_i^- & \text{ for }i \in \{1, \dots, n\} \setminus J_k\\
  c_k^+ / a_k \cdot a_{j_k} & \text{ for }i = j_k\;,
\end{array}\end{cases}\]
Note that $c_{j_k}^- \leq c_k^+ / a_k \cdot a_{j_k} \leq c_{j_k}^+$
holds because by construction $p_{j_k}^- \leq p_k^- \leq p_{j_k}^+$ as
$j_k \in I_k^0$.  Indeed, this definition ensures that the items $i
\in J_k \setminus \{j_k\}$ precede $j_k$ in the follower's sorting,
since $p_i^- \leq p_k^-$ and therefore, $c_i / a_i = c_i^+ / a_i = -
p_i^- \geq - p_k^- = c_k^+/a_k = c_{j_k} / a_{j_k}$. Analogously, the
items $i \in \{1, \dots, n\} \setminus J_k$ will be packed after $j_k$
by the follower.

Note that this solution sets each variable except for $c_{j_k}$ to an
endpoint of its corresponding interval. In general, there is no
optimal solution with all variables set to an interval endpoint. This
can be seen by the following example.

\begin{example}
  Let $n = 3$ and $a=(1,1,1)^\top$, $b =
\frac{3}{2}$, $U = \{3\} \times \{2\} \times [1, 4]$, $d = (-1,
1, 0)^\top$. The optimal solution returned by the algorithm is $(\{1,
3\}, 3, \frac{1}{2})$ with value $1$. For the follower to select the
first item and half of the third item, i.e., for~$c_1 \geq c_3
\geq c_2$ to hold, the adversary must choose $c_3 \in [2, 3]$, so it
cannot be at one of the endpoints of the interval $[1, 4]$.
\end{example}

\subsection{Solving the leader's problem}

Next, we describe an algorithm to solve the robust bilevel
optimization problem, which performs the maximization over the
capacity~$b$. For this, we will use the variant of Dantzig's algorithm
which returns a piecewise linear function, as described in
Section~\ref{section_withoutuncertainty}. We call this routine
\textsc{BilevelDantzig} and assume its input to be~$I \subseteq \{1,
\dots, n\}$, $d \in \mathbb{R}^{n}$, $a \in \mathbb{R}_{> 0}^{n}$, and~$0 \leq b^- \leq b^+ \leq \sum_{i \in I} a_i$, since, as before, we
need it also to work on a subset of the item set. The output is a
piecewise linear function~$f\colon [b^-, b^+] \to \mathbb{R}_{\geq
  0}$, which can be represented by a list of all its vertices, given
as points of the graph of~$f$.  The leader's problem can now be solved
by Algorithm~\ref{algorithm_leader}.

\begin{algorithm}
  \caption{Algorithm for the leader's problem}
  \label{algorithm_leader}

  \Input{~$a \in \mathbb{R}_{> 0}^n$, $0 \leq b^- \leq b^+ \leq \sum_{i
      = 1}^n a_i$, $d \in \mathbb{R}^n$, $p^-, p^+ \in \mathbb{R}_{<
      0}^n$ with $p^- \leq p^+$,\\ inducing an interval order $P$}

  \Output{~value $b\in [b^-,b^+]$ maximizing the result of
    Algorithm~\ref{algorithm_adversary}}

  $K := \emptyset$

  \For{$k = 1, \dots, n$}{
    $I_k^- := \{i \in \{1, \dots, n\} \colon p_i^+ < p^-_k\}$

    $I_k^0 := \{i \in \{1, \dots, n\} \colon p_i^- \leq p_k^- \leq
    p_i^+\}$

    \If{$0 \leq b^+ - \sum_{i \in I_k^-} a_i$ and $b^- - \sum_{i \in
        I_k^-} a_i \leq \sum_{i \in I_k^0}
      a_i$}{ \label{line_rangecondition2}

      $\tilde{b}^- := \max\{0, b^- - \sum_{i \in I_k^-} a_i\}$

      $\tilde{b}^+ := \min\{b^+ - \sum_{i \in I_k^-} a_i, \sum_{i \in
        I_k^0} a_i\}$
      
      $f_k' := \textsc{BilevelDantzig}(I_k^0, d, a, \tilde{b}^-,
      \tilde{b}^+)$ \label{line_DantzigIk02}

      $f_k := f_k' + (\sum_{i \in I_k^-} a_i, \sum_{i \in I_k^-} d_i)$
      \hfill // shift all vertices of~$f_k'$ by~$\sum_{i \in I_k^-} a_i$ in $x$

      \hfill // and by $\sum_{i \in I_k^-} d_i$ in $y$~direction

      $K := K \cup \{k\}$
    }
  }

  $f := \min\{f_k \colon k \in K\}$ \hfill // pointwise minimum \label{line_pointwisemin}

  \Return $\text{argmax}\{f(b) \colon b^- \leq b \leq b^+\}$ \label{line_max}
\end{algorithm}

\begin{lemma}
  Algorithm~\ref{algorithm_leader} is correct.
\end{lemma}

\begin{proof}
  First note that Algorithm~\ref{algorithm_adversary} can be
  considered as the special case of Algorithm~\ref{algorithm_leader}
  where~$b = b^- = b^+$. For the correctness of
  Algorithm~\ref{algorithm_leader}, it suffices to show that the
  function $f$ describes the value of the output of
  Algorithm~\ref{algorithm_adversary} depending on $b \in [b^-,
  b^+]$. For $b = \sum_{i = 1}^n a_i$, which is only possible if~$b^+
  = \sum_{i = 1}^n a_i$, consider any iteration $k = k^*$, where
  $k^* \in \{1, \dots, n\}$ is chosen such that $p_{k^*}^-$ is
  maximum. It is easy to check that this iteration results in a
  function $f_{k^*}$ with $f_{k^*}(b^+) = \sum_{i = 1}^n d_i$,
  corresponding to the fractional prefix $\bar{F}$ returned in
  Algorithm~\ref{algorithm_adversary}, and that for no $k$ with $p_k^-
  < p_{k^*}^-$, the function $f_k$ is defined for $b$. For the
  rest of the proof, we assume $b < \sum_{i = 1}^n a_i$.

  The condition in line~\ref{line_rangecondition2} ensures
  that~$\tilde{b}^-\le \tilde{b}^+$, so that the call of
  \textsc{BilevelDantzig} in line~\ref{line_DantzigIk02} is valid.
  Let $(J_k^b, j_k^b, \lambda_k^b)$ be the fractional prefix $(J_k,
  j_k, \lambda_k)$ in Algorithm~\ref{algorithm_adversary} called for
  $b \in [b^-, b^+]$ with $b < \sum_{i = 1}^n a_i$.  We claim that,
  for all~$k \in \{1, \dots, n\}$ and all~$b \in [b^-, b^+]$ with~$b <
  \sum_{i = 1}^n a_i$, the function $f_k$ in
  Algorithm~\ref{algorithm_leader} is defined in the point $b$ if and
  only if $(J_k^b, j_k^b, \lambda_k^b)$ is defined in
  Algorithm~\ref{algorithm_adversary}, and then $f_k(b) = d^\top
  x^{(J_k^b, j_k^b, \lambda_k^b)}$.

  Let $k \in \{1, \dots, n\}$ and $b \in [b^-, b^+]$ with $b < \sum_{i
    = 1}^n a_i$. Then $f_k(b)$ is defined if and only if~$\tilde{b}^-
  + \sum_{i \in I_k^-} a_i \leq b \leq \tilde{b}^+ + \sum_{i \in
    I_k^-} a_i$, i.e., if and only if
  \[\textstyle\sum_{i \in I_k^-} a_i \leq b
  \leq \sum_{i \in I_k^0} a_i + \sum_{i \in I_k^-} a_i\;,\] which is
  almost the same condition as the one for defining $(J_k^b, j_k^b,
  \lambda_k^b)$. Actually, $(J_k^b, j_k^b, \lambda_k^b)$ is not
  defined if $b = \sum_{i \in I_k^0} a_i + \sum_{i \in I_k^-} a_i$,
  but this is only for convenience in the formulation of
  Algorithm~\ref{algorithm_adversary}. We could define it as
  $(I_k^- \cup I_k^0 \cup \{j^*\}, j^*, 0)$, where $j^* \in \{1,
  \dots, n\} \setminus (I_k^- \cup I_k^0)$ is chosen such that
  $p_{j^*}^-$ is minimum. However, this is not relevant since this
  fractional prefix is also considered in iteration $j^*$. Note that
  such a $j^*$ always exists because of~$b < \sum_{i = 1}^n
  a_i$.

  In case~$f_k(b)$ is defined, the corresponding values $f_k(b)$ and
  $d^\top x^{(J_k^b, j_k^b, \lambda_k^b)}$ agree because the piecewise
  linear function returned by \textsc{BilevelDantzig} consists of the
  values of the solutions returned by \textsc{Dantzig} for given
  values of $b$. Since $f(b)$ is the minimum of all $f_k(b)$ that are
  defined, this implies that it is equal to the value of the optimal
  solution computed in Algorithm~\ref{algorithm_leader}.
\end{proof}

\begin{theorem}
  The robust bilevel continuous knapsack problem with interval
  uncertainty can be solved in $\mathcal{O}(n^2 \log n)$ time.
\end{theorem}

\begin{proof}
  In every of the $n$ iterations, Algorithm~\ref{algorithm_leader}
  needs $\mathcal{O}(n)$ time to compute the sets~$I_k^-$ and~$I_k^0$,
  and, since $|I_k^0| \leq n$, $\mathcal{O}(n \log n)$ time for
  Dantzig's algorithm. As explained in
  Section~\ref{section_scenariouncertainty}, the pointwise minimum of
  the at most $n$ piecewise linear functions with at most $n$ segments
  each, as well as the maximum of the resulting function
  (lines~\ref{line_pointwisemin} and~\ref{line_max}) can be computed
  in $\mathcal{O}(n^2 \log n)$ time.  
\end{proof}

\begin{remark} \label{rem_interval}
  If there are intersections of intervals only consisting in one
  point, no special treatment is needed in the pessimistic problem
  version since the worst possible order of the corresponding items
  from the leader's perspective will be chosen anyway. If we use the
  optimistic approach, we have to be more careful about that. However,
  if no one-point intersections occur, there is no difference between
  the optimistic and the pessimistic case regarding the sortings the
  adversary can enforce. Hence, Algorithm~\ref{algorithm_adversary}
  and Algorithm~\ref{algorithm_leader} can be used for the optimistic
  case without any changes.

  In the optimistic setting with one-point intersections, the set
  of follower's sortings the adversary can enforce does not even have
  to be the set of linear extensions of a partial order. However,
  Algorithm~\ref{algorithm_adversary} and
  Algorithm~\ref{algorithm_leader} can be modified to work also in the
  general optimistic setting as follows. Given an item $k \in \{1,
  \dots, n\}$, we distinguish two cases:
  \begin{enumerate}
  \item If $p_k^- < p_k^+$, the follower will sort every item $i \in
    \{1, \dots, n\}$ with $p_i^+ = p_k^-$ before item~$k$, given an
    optimal adversary's choice. Indeed, the adversary either prefers
    the order where $i$ precedes $k$ anyway (i.e., $d_i / a_i < d_k /
    a_k$) and can enforce it, e.g., by choosing $p_k = p_k^- +
    \varepsilon$ for some $\varepsilon > 0$ that is small enough, or
    he prefers the order where $k$ precedes~$i$ (i.e., $d_i / a_i >
    d_k / a_k$), but cannot enforce it because even if choosing $p_i =
    p_k = p_k^-$, the follower will choose item $i$ before item $k$
    because of the optimism. If $d_i / a_i = d_k / a_k$ holds, the
    leader and the adversary are indifferent about the order of $i$
    and $k$, and we may assume that $i$ precedes $k$ then, as
    well. With this idea, one can show that setting
    \begin{eqnarray*}
     I_k^- & := & \{i \in \{1, \dots, n\} \colon p_i^+ \leq p_k^-\}\\
     I_k^0 & := & \{i \in \{1, \dots, n\} \colon p_i^- \leq p_k^- < p_i^+\}
    \end{eqnarray*}
    in Algorithm~\ref{algorithm_adversary}
    and Algorithm~\ref{algorithm_leader} for this $k$ leads to the
    desired result.
  \item If $p_k^- = p_k^+$, the follower will sort every item $i \in
    \{1, \dots, n\}$ with $p_i^- < p_i^+ = p_k^-$ before item $k$,
    given an optimal adversary's choice, by similar arguments as in
    the first case. Again similarly, every item $i \in \{1, \dots,
    n\}$ with $p_i^- = p_k^- < p_i^+$ will be chosen after item~$k$. Among the items $i \in \{1, \dots, n\}$ with $p_i^- = p_i^+ =
    p_k^-$, the adversary has no choice at all, and the follower will
    choose the best order from the leader's perspective, i.e., every
    such~$i$ with $d_i / a_i > d_k / a_k$ will precede $k$, while
    every $i$ with $d_i / a_i < d_k / a_k$ will be chosen after~$k$. We may assume that there are no such items $i$ with $d_i /
    a_i = d_k / a_k$ because they could be merged into a single item
    together with item $k$ in the beginning. With this knowledge, we
    can define
    \begin{eqnarray*}
      I_k^- & := & \{i \in \{1, \dots, n\} \colon p_i^- < p_i^+ \leq p_k^- \text{ or } (p_i^- = p_i^+ =
      p_k^- \text{ and } d_i / a_i > d_k / a_k)\}\\
      I_k^0 & := & \{k\} \cup \{i \in \{1, \dots, n\} \colon p_i^- < p_k^- < p_i^+\}
    \end{eqnarray*}
    in iteration $k$ of Algorithm~\ref{algorithm_adversary} and
    Algorithm~\ref{algorithm_leader}. It might happen that there is no
    adversary's choice leading to the follower choosing the desired
    fractional prefix $(J_k, j_k, \lambda_k)$ if $j_k \neq k$ and
    there is a $j^* \in \{1, \dots, n\}$ with $p_{j^*}^- = p_{j^*}^+ =
    p_k^-$ and $d_{j^*} / a_{j^*} < d_k / a_k$. But in this case,
    iteration $k$ is not relevant anyway, hence we can just omit $k$
    from $K$ if this happens. Indeed, given an optimal solution in
    which $p_k^- = p_k^+$ holds for every possible head $k$, one can
    show that the fractional item of the optimal solution can be
    assumed to be the head $k$ that minimizes $d_i/a_i$.
  \end{enumerate}
\end{remark}

\section{Discrete uncorrelated uncertainty} \label{section_discreteuncorrelateduncertainty}

We now consider the robust bilevel continuous knapsack problem with an
uncertainty set of the form $U = U_1 \times \dots \times U_n$, where
each $U_i$ is a finite set, i.e., for every component of $c$, there is
a finite number of options the adversary can choose from and these are
independent of each other. It turns out that this version of the
problem is NP-hard in general, even if we consider the special case $U
= \{c_1^-, c_1^+\} \times \dots \times \{c_n^-, c_n^+\}$, where there
are only two options in every component.

\begin{theorem} \label{thm_discreteuncorrelated}
  The robust bilevel continuous knapsack problem with an uncertainty
  set being the product of finite sets, is NP-hard, even if each of
  these sets has size two.
\end{theorem}

We show Theorem~\ref{thm_discreteuncorrelated} by a reduction from the
well-known NP-hard subset sum problem. Let $m \in \mathbb{N}$ and
$w_1, \dots, w_m, W \in \mathbb{N}$ be an instance of subset
sum. Without loss of generality, we may assume $1 \leq W \leq \sum_{i
  = 1}^m w_i - 1$. We show that we can decide if there is a subset~$S
\subseteq \{1, \dots, m\}$ with $\sum_{i \in S} w_i = W$ in polynomial
time if the following instance of the robust bilevel continuous
knapsack problem can be solved in polynomial time: define
$\varepsilon:=\frac{1}{4}$ and~$M:=\sum_{i = 1}^m w_i + \varepsilon$.
Let $n := m + 2$ and set
\begin{eqnarray*}
  a &:=& (\varepsilon, w_1, \dots, w_m, M)^\top\\
  d &:=& (-M, - w_1, \dots, - w_m,\varepsilon)^\top\\
  b^- &:=& W\\
  b^+ &:=& W + 2 \varepsilon\;.
\end{eqnarray*}
The uncertainty set is defined as
$$U := \{1 \cdot a_1, (n + 1) \cdot a_1\} \times \{2 \cdot a_2, (n + 2) \cdot a_2\}
\times \dots \times \{n \cdot a_n, 2n \cdot a_n\}\;,$$ which leads to
$c_1^-/a_1 < \dots < c_n^-/a_n < c_1^+/a_1 < \dots < c_n^+/a_n$. Note
that, since all these values are distinct, the optimistic and
pessimistic approach do not have to be distinguished for this
instance.

In the following two lemmas, we investigate the structure of optimal
follower's solutions and of the leader's objective function,
respectively. Both are done for a large range $[\varepsilon, M)$ of
values for $b$, which will be useful to understand the two different
behaviors the function can have in the actual range $[b^-, b^+]$
afterwards, depending on the subset sum instance being a yes or a no
instance.

\begin{lemma}\label{lemma_binary}
  Let any leader decision $b \in [\varepsilon, M)$ be given. Then, for
  every optimal choice of the adversary, the resulting follower's
  solution~$x$ satisfies $x_i\in\{0,1\}$ for all~$i \in
  \{2,\dots,n-1\}$.
\end{lemma}

\begin{proof}
  First note that the adversary can always enforce a solution with a
  value of at most~$-M$, e.g., by setting $c_1 = c_1^+$ and $c_i =
  c_i^-$ for all $i \in \{2, \dots, n\}$, so that the first item is
  taken first and therefore completely, i.e., $x_1 = 1$, because
  $a_1=\varepsilon \leq b$. Since every solution with $x_1 = 0$ has a
  value greater than $-M$, this implies $x_1 > 0$ for every follower's
  solution resulting from an optimal choice of the
  adversary. Moreover, note that we always have $x_n < 1$ as~$a_n = M
  > b$.

  Now let~$i\in\{2,\dots,n-1\}$. If the adversary chooses $c_i=c_i^+$,
  it follows that $x_i=1$. Indeed,~$x_i<1$ would imply $x_1=0$ since
  $c_1/a_1 < c_i^+/a_i$ always holds. Analogously, if the adversary
  chooses $c_i=c_i^-$, we have~$x_i=0$, as~$x_i>0$ would imply
  $x_n=1$, since $c_i^-/a_i < c_n/a_n$. Therefore, an optimal
  adversary's solution always leads to a follower's solution $x$ such that~$x_i \in \{0, 1\}$ for all $i \in \{2, \dots, n - 1\}$. 
\end{proof}

For the remainder of the proof, denote the leader's objective function
by~$f$. This function is described for all $b \in [\varepsilon, M)$ by
the following

\begin{lemma}\label{lemma_shape_f}
  Let $V_1,V_2\in\mathbb{N}_0$ be two values with $V_1 < V_2$ that can
  arise as sums of subsets of~$\{w_1,\dots,w_m\}$, such that there is
  no other such value $V$ with $V_1 < V < V_2$. Then
  $$f(b)=\begin{cases}
    \begin{array}{ll}
      -M-V_1+\frac{\varepsilon}{M}(b-V_1-\varepsilon), & ~\text{if }b\in[V_1+\varepsilon,V_2)\\
      \min\left\{-M-V_1+\frac{\varepsilon}{M}(b-V_1-\varepsilon),-\frac
        M\varepsilon(b-V_2)-V_2\right\}, &
      ~\text{if }b\in[V_2,V_2+\varepsilon)\;.
    \end{array}
  \end{cases}$$
\end{lemma}

\begin{proof}
  Let~$\sum_{i \in T} w_i = V_1$ for some~$T\subseteq \{1, \dots,
  m\}$. Using Lemma~\ref{lemma_binary}, it is easy to see that for
  all~$b\in [V_1+\varepsilon,V_2)$, the unique best choice of the
  adversary is to let the follower set~$x_1=1$ and pack~$2,\dots,n-1$
  according to~$T$, the rest is filled by a fraction of
  item~$n$. Indeed, for~$b<V_2$, no better packing of
  items~$2,\dots,n-1$ than~$T$ is possible, and since~$b\ge
  V_1+\varepsilon$, the most profitable item~1 (from the adversary's
  perspective) can be added entirely without making the packing~$T$
  infeasible. The adversary can produce this solution by
  choosing~$c_1=c_1^+$, $c_{n}=c_n^-$ and for $i \in \{2, \dots, n -
  1\}$, $c_i = c_i^+$ if $i - 1 \in T$, or $c_i = c_i^-$ otherwise.
  This leads to
  $$\textstyle f(b)=-M-V_1+\frac{\varepsilon}{M}(b-V_1-\varepsilon)$$
  for $b \in [V_1+\varepsilon,V_2)$.  If~$b\in [V_2,V_2+\varepsilon)$,
  the same solution, with a larger fraction of item~$n$, is still
  possible. However, since now a better packing $S \subseteq \{1,
  \dots, m\}$ with $\sum_{i \in S} w_i = V_2$ is available, we also
  have to consider to pack according to~$S$ and add as much of
  item~$1$ as possible (which the adversary can obtain again by
  setting~$c_1=c_1^+$, $c_{n}=c_n^-$ and, for all~$i \in \{2, \dots, n -
  1\}$, $c_i = c_i^+$ if $i - 1 \in S$, or $c_i = c_i^-$
  otherwise). This solution has value
  $$-\frac M\varepsilon(b-V_2)-V_2\;,$$
  and the adversary chooses the solution leading to a smaller value
  of~$f$.  
\end{proof}

Let~$b^* \in [b^-, b^+] = [W, W + 2 \varepsilon]$ be an optimal
leader's solution in the original range. We conclude the proof of
Theorem~\ref{thm_discreteuncorrelated} by showing that the given
instance of subset sum is a yes instance if and only if~$b^*\neq b^+$,
which follows from the following two lemmas. The two situations
arising in Lemma~\ref{lemma_no} and Lemma~\ref{lemma_yes} are
illustrated in Figure~\ref{figure_np}.

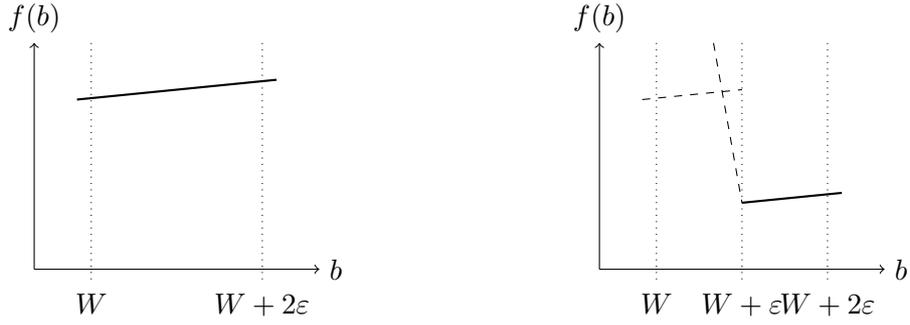
\begin{figure}
\begin{subfigure}[c]{0.5\textwidth}
  \begin{center}
    \begin{tikzpicture}[scale=0.75]
      \draw[->] (0,0) -- (5,0) node[anchor=west] {$b$};
      \draw[->] (0,0) -- (0,4) node[anchor=south] {$f(b)$};
      \draw[dotted] (1,4) -- (1,-0.25) node[anchor=north] {$W$};
      \draw[dotted] (4,4) -- (4,-0.25) node[anchor=north] {$W+2\varepsilon$};
      \draw[thick] (0.75,3) -- (4.25,3.35);
    \end{tikzpicture}
  \end{center}
  \subcaption{The case of a no instance (Lemma~\ref{lemma_no})}
\end{subfigure}
\begin{subfigure}[c]{0.5\textwidth}
  \begin{center}
    \begin{tikzpicture}[scale=0.75]
      \draw[->] (0,0) -- (5,0) node[anchor=west] {$b$};
      \draw[->] (0,0) -- (0,4) node[anchor=south] {$f(b)$};
      \draw[dotted] (1,4) -- (1,-0.25) node[anchor=north] {$W$};
      \draw[dotted] (2.5,4) -- (2.5,-0.25) node[anchor=north] {$W+\varepsilon$};
      \draw[dotted] (4,4) -- (4,-0.25) node[anchor=north] {$W+2\varepsilon$};
      \draw[thick] (2.5,1.35-0.175) -- (4.25,1.35);
      \draw[dashed] (2,4) -- (2.5,1.35-0.175);
      \draw[dashed] (0.75,3) -- (2.5,3.175);
    \end{tikzpicture}
  \end{center}
  \subcaption{The case of a yes instance (Lemma~\ref{lemma_yes})}
\end{subfigure}
\caption{The leader's objective function~$f$}\label{figure_np}
\end{figure}

\begin{lemma}\label{lemma_no}
  If the given instance~$(w_1,\dots,w_m,W)$ of subset sum is a no
  instance, then~$b^* = b^+$.
\end{lemma}

\begin{proof}
  Let $V_1$ and $V_2$ be the largest subset sum with $V_1 < W$ and the
  smallest subset sum with $V_2 > W$, respectively. Then we have
  $[b^-, b^+] = [W, W + 2 \epsilon] \subseteq [V_1 +
  \varepsilon,V_2)$. Hence, by Lemma~\ref{lemma_shape_f}, the
  function~$f$ is linear on~$[b^-,b^+]$ with
  slope~$\frac{\varepsilon}{M}>0$. 
\end{proof}

\begin{lemma}\label{lemma_yes}
  If the given instance~$(w_1,\dots,w_m,W)$ of subset sum is a yes
  instance, then $f(b^-)>f(b^+)$ and hence~$b^*\neq b^+$.
\end{lemma}

\begin{proof}
  Let~$V$ be the largest subset sum with~$V<W$. We obtain~$f(b^+)=f(W
  + 2 \varepsilon) = -M-W+\tfrac{\varepsilon^2}M$ by applying
  Lemma~\ref{lemma_shape_f} to~$V_1=W$.  By Lemma~\ref{lemma_shape_f}
  applied to~$V_1=V$ and~$V_2=W$, the two possible values of~$f$
  in~$b^- = W$ are $-M-V+\frac{\varepsilon}{M}(W-V-\varepsilon)$ and
  $-W$.  Both can be easily seen to be strictly larger than~$f(b^+)$.
\end{proof}

This concludes the proof of
Theorem~\ref{thm_discreteuncorrelated}. Note that this also proves the
NP-hardness of the adversary's problem:

\begin{theorem}
  Evaluating the objective function value of the robust bilevel
  continuous knapsack problem with an uncertainty set being the
  product of finite sets, in some feasible point~$b$, is NP-hard, even
  if each of the sets has size two.
\end{theorem}

\begin{proof}
  Looking at the proof of Theorem~\ref{thm_discreteuncorrelated}, the
  two cases in Lemma~\ref{lemma_no} and Lemma~\ref{lemma_yes} can also
  be distinguished by computing the adversary's optimal solution value
  (or the leader's objective function value, equivalently) in the
  point $b = b^+ = W + 2 \varepsilon$. In case of a yes instance it
  is $$f(b) = -M - W + \frac{\varepsilon}{M} (b - W - \varepsilon),$$
  while in case of a no instance it is $$f(b) = -M - V +
  \frac{\varepsilon}{M} (b - V - \varepsilon),$$ where $V$ is the
  largest subset sum with $V < W$, which leads to the second value
  being larger than the first one. This shows that there is also a
  reduction from the subset sum problem to the adversary's problem.
\end{proof}

In Section~\ref{section_intervaluncertainty}, we have seen that the
robust bilevel continuous knapsack problem can be solved efficiently
when each coefficient is chosen independently from a given
interval. Theorem~\ref{thm_discreteuncorrelated} shows that the same
problem turns NP-hard when the adversary is only allowed to choose the
follower's objective coefficients from the endpoints of the
intervals. In particular, this implies that replacing an uncertainty
set by its convex hull may change the problem significantly, in
contrast to the situation in single-level robust optimization.
This can also be shown by the following explicit example.

\begin{example}\label{ex}
  Let $n=5$ and define $a = (1, 1, 1, 1, 1)^\top$, $b^- = 0$, $b^+ =
  5$, $d = (2,-1,1,-2,0)^\top$, and $U = \{5\} \times \{4\} \times \{3\} \times \{2\} \times \{1,
  6\}$. In this instance, the order of
  the items $1$, $2$, $3$, and $4$ is fixed, while item $5$ could be
  in the first or last position with respect to uncertainty set $U$,
  but also in every position in between the other items when the
  uncertainty set is~$\text{conv}(U)=\{5\} \times \{4\} \times \{3\}
  \times \{2\} \times [1, 6]$. The leader's objective function
  on~$[0,5]$ is depicted in Figure~\ref{fig_ex}, for uncertainty
  sets~$U$ and~$\text{conv}(U)$.

  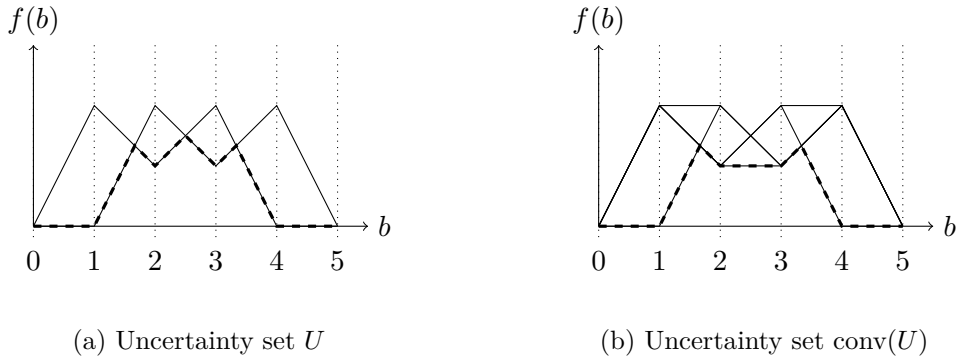
\begin{figure}
    \begin{subfigure}[c]{0.5\textwidth}
      \begin{center}
        \begin{tikzpicture}[scale=0.8]
          \draw[->] (0,0) -- (5.5,0) node[anchor=west] {$b$};
          \draw[->] (0,0) -- (0,3) node[anchor=south] {$f(b)$};
          \draw[dotted] (0,3) -- (0,-0.25) node[anchor=north] {$0$};
          \draw[dotted] (1,3) -- (1,-0.25) node[anchor=north] {$1$};
          \draw[dotted] (2,3) -- (2,-0.25) node[anchor=north] {$2$};
          \draw[dotted] (3,3) -- (3,-0.25) node[anchor=north] {$3$};
          \draw[dotted] (4,3) -- (4,-0.25) node[anchor=north] {$4$};
          \draw[dotted] (5,3) -- (5,-0.25) node[anchor=north] {$5$};
          \draw (0,0) -- (1,2) -- (2,1) -- (3,2) -- (4,0) -- (5,0);
          \draw (0,0) -- (1,0) -- (2,2) -- (3,1) -- (4,2) -- (5,0);
          \draw[very thick, dashed] (0,0) -- (1,0) -- (5/3,3-5/3) -- (2,1) -- (5/2,3/2) -- (3,1) -- (5-5/3,3-5/3) -- (4,0) -- (5,0);
        \end{tikzpicture}
      \end{center}
      \subcaption{Uncertainty set $U$}
    \end{subfigure}
    \begin{subfigure}[c]{0.5\textwidth}
      \begin{center}
        \begin{tikzpicture}[scale=0.8]
          \draw[->] (0,0) -- (5.5,0) node[anchor=west] {$b$};
          \draw[->] (0,0) -- (0,3) node[anchor=south] {$f(b)$};
          \draw[dotted] (0,3) -- (0,-0.25) node[anchor=north] {$0$};
          \draw[dotted] (1,3) -- (1,-0.25) node[anchor=north] {$1$};
          \draw[dotted] (2,3) -- (2,-0.25) node[anchor=north] {$2$};
          \draw[dotted] (3,3) -- (3,-0.25) node[anchor=north] {$3$};
          \draw[dotted] (4,3) -- (4,-0.25) node[anchor=north] {$4$};
          \draw[dotted] (5,3) -- (5,-0.25) node[anchor=north] {$5$};
          \draw (0,0) -- (1,2) -- (2,1) -- (3,2) -- (4,0) -- (5,0);
          \draw (0,0) -- (1,0) -- (2,2) -- (3,1) -- (4,2) -- (5,0);
          \draw (0,0) -- (1,2) -- (2,1) -- (3,2) -- (4,2) -- (5,0);
          \draw (0,0) -- (1,2) -- (2,1) -- (3,1) -- (4,2) -- (5,0);
          \draw (0,0) -- (1,2) -- (2,2) -- (3,1) -- (4,2) -- (5,0);
          \draw[very thick, dashed] (0,0) -- (1,0) -- (5/3,3-5/3) -- (2,1) -- (3,1) -- (5-5/3,3-5/3) -- (4,0) -- (5,0);
        \end{tikzpicture}
      \end{center}
      \subcaption{Uncertainty set $\text{conv}(U)$}
    \end{subfigure}
    \caption{The leader's objective function in Example~\ref{ex}}\label{fig_ex}
  \end{figure}

  In the former case, the leader's unique optimal solution
  is~$b=\tfrac 52$ with objective value~$\tfrac 32$, while in the
  latter case the two optimal solutions are~$b=\tfrac 53$ and $b=\tfrac{10}3$
  with objective value~$\tfrac 43$.
\end{example}

\section{Simplicial uncertainty}\label{section_simplicialuncertainty}

We next consider uncertainty sets being simplices and again show that
the problem is NP-hard in this case. This case is of interest because
it can be considered a special case of two commonly used types
of uncertainty in robust optimization: polytopal and
Gamma-uncertainty. Moreover, it can be viewed as the convex variant of the
discrete uncertainty case investigated in Section~\ref{section_scenariouncertainty}.
We will write the uncertainty set in the following as
$$\textstyle U_{\hat c,\Gamma} = \{c \in \mathbb{R}^n \colon c_i \geq \hat{c}_i \text{ for all } i \in \{1, \dots, n\}, \sum_{i = 1}^n (c_i - \hat{c}_i) \leq \Gamma\},$$
where a vector $\hat{c} \in \mathbb{R}^n_{> 0}$ and a number $\Gamma >
0$ bounding the deviation from $\hat{c}$ are given. We have

\begin{theorem} \label{thm_simplex}
  The robust bilevel continuous knapsack problem with simplicial
  uncertainty set~$U_{\hat c,\Gamma}$, where $\hat{c}$ and $\Gamma$
  are part of the input, is NP-hard.
\end{theorem}

\begin{proof}
  We again show this by a reduction from the subset sum problem. Let
  $m \in \mathbb{N}$ and $w_1, \dots, w_m, W \in \mathbb{N}$ be an
  instance of subset sum. We show that we can decide if there is a
  subset $S \subseteq \{1, \dots, m\}$ with $\sum_{i \in S} w_i = W$
  in polynomial time if the following instance of the robust bilevel
  continuous knapsack problem can be solved in polynomial time:
  define~$n:=m+1$, $M:=\sum_{i = 1}^m w_i + 1$ and
  \begin{eqnarray*}
    a &:=& (w_1, \dots, w_m, M)^\top\\
    d &:=& (-w_1, \dots, -w_m, M)^\top\\
    b^- &:=& 0\\
    b^+ &:=& \sum_{i = 1}^n a_i = \sum_{i = 1}^{m}w_i + M\\
    \hat{c} &:=& ((2M - 1) w_1, \dots, (2M - 1) w_m, 2M^2)^\top\\
    \Gamma &:=& W\;.
  \end{eqnarray*}
  As in the discrete uncertainty case, we may imagine the leader's
  objective function~$f$ as the pointwise minimum of all piecewise
  linear functions~$f_c$ for~$c \in U_{\hat c,\Gamma}$. Note that,
  although there are infinitely many $c \in U_{\hat c,\Gamma}$, the
  number of distinct functions $f_c$ is finite since the functions
  depend only on the follower's sorting of the items induced by
  $c$. Since $d_i/a_i = -1$ for all $i = 1, \dots, m$ and $d_{m +
    1}/a_{m + 1} = 1$, all functions~$f_c$ have the structure shown in
  Figure~\ref{figure_np2}, where $V_c := \sum_{i \in I_c} a_i =
  \sum_{i \in I_c} w_i$ for
  $$I_c := \{i \in \{1, \dots, m\} \colon c_i/a_i \geq
  c_{m+1}/a_{m+1}\}.$$ The leader's objective function is the
  pointwise minimum of the functions~$f_c$. It follows easily that it
  agrees with~$f_{c^*}$ where~$c^*\in\argmax_{c\in U_{\hat
      c,\Gamma}}V_c$. In particular, since
  $$f(V_{c^*} + M) = -V_{c^*}
  + M > 0 = f(0)\;,$$
  the leader's optimal solution is~$V_{c^*}+M$ with
  value~$-V_{c^*}+M$, so that it remains to show that by
  computing~$V_{c^*}$ we can decide the subset sum instance.
  \begin{figure}
    \begin{center}
      \begin{tikzpicture}[scale=0.65]
        \draw[->] (0,0) -- (8,0) node[anchor=west] {$b$};
        \draw[->] (0,-2) -- (0,4) node[anchor=south] {$f_c(b)$};
        \draw[dotted] (1,4) -- (1,-1.25) node[anchor=north] {$V_c$};
        \draw[dotted] (5,4) -- (5,-1.25) node[anchor=north] {$V_c+M$};
        \draw[dotted] (7,4) -- (7,-1.25) node[anchor=north] {$b^*$};
        \draw[dotted] (8,-1) -- (-0.25,-1) node[anchor=east] {$-V_c$};
        \draw[dotted] (8,3) -- (-0.25,3) node[anchor=east] {$-V_c+M$};
        \draw[dotted] (8,1) -- (-0.25,1) node[anchor=east] {$-\sum_{i = 1}^m w_i + M$};
        \draw[thick] (0,0) -- (1,-1) -- (5,3) -- (7, 1);
      \end{tikzpicture}
    \end{center}
    \caption{The function~$f_c$}\label{figure_np2}
  \end{figure}
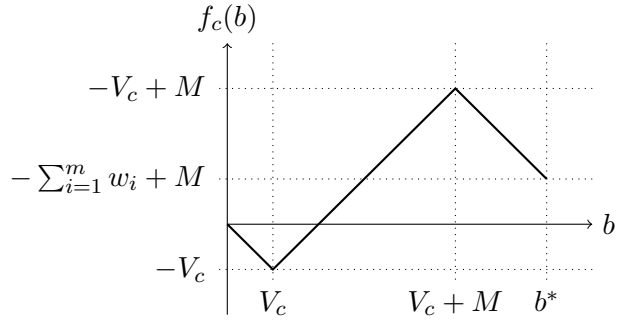

  Since $\hat{c}_i/a_i = 2M - 1 < 2M = \hat{c}_{m + 1}/a_{m + 1}$ for
  all $i = 1, \dots, m$, the adversary has two options for every item
  $i$:
  \begin{enumerate}
  \item Either he decides to include $i$ in the set $I_c$, which can
    be achieved by shifting the value $\hat{c}_i/a_i$ to the right by
    $1$ up to $\hat{c}_{m + 1}/a_{m + 1}$, by choosing
    $c_i = \hat{c}_i + a_i = \hat{c}_i + w_i$ and thus paying
    $c_i - \hat{c}_i = w_i$ in the constraint
    $\sum_{i = 1}^m (c_i - \hat{c}_i) \leq \Gamma$ on the uncertainty
    set. Note that there is no incentive to choose $c_i$ greater than
    that, as it would not change the leader's objective function value.
  \item Or he decides not to include $i$ in the set $I_c$, for which
    it is most efficient in terms of the $\Gamma$-constraint to choose
    $c_i = \hat{c}_i$.
  \end{enumerate}
  Note that there is no incentive to choose $c_{m + 1} > \hat{c}_{m +
    1}$. Therefore, we may assume that $c_i \in \{\hat{c}_i, \hat{c}_i
  + w_i\}$ for all $i \in \{1, \dots, m\}$ and $c_{m + 1} = \hat{c}_{m
    + 1}$ hold in an optimal adversary's solution. By
  maximizing~$V_c=\sum_{i \in I_c} w_i$, we can thus compute the
  largest subset sum~$V$ with $V \leq W$, since $\sum_{i \in I_c} w_i
  = \sum_{i = 1}^m (c_i - \hat{c}_i) \leq \Gamma = W$. Thus, the
  subset sum instance is a yes instance if and only if $V_{c^*} = W$.
\end{proof}

\begin{remark}\label{rem_simplex}
  In the proof of Theorem~\ref{thm_simplex}, we used the assumption of
  the pessimistic approach, because item~$i$ was included in the
  set~$I_c$ even if $c_i / a_i = c_{m + 1} / a_{m + 1}$. In the
  optimistic approach, item~$i$ may be packed only if strict
  inequality holds in the definition of~$I_c$. Without changing the
  effect of the constraint $\sum_{i = 1}^m (c_i - \hat{c}_i) \leq
  \Gamma$, this could be modeled by choosing~$\hat{c}_{m + 1}$
  slightly smaller than~$2 M^2$ in the definition of the instance.
\end{remark}

\pagebreak[3]
\noindent
Again, also the adversary's problem is NP-hard:

\begin{theorem}\label{theorem_simplex_eval}
  Evaluating the objective function value of the robust bilevel
  continuous knapsack problem with simplicial uncertainty set $U_{\hat
    c,\Gamma}$, where $\hat{c}$ and $\Gamma$ are part of the input, in
  some feasible point $b$, is NP-hard.
\end{theorem}

\begin{proof}
  Looking at the proof of the previous theorem, it can be seen that
  the subset sum instance is a yes instance if and only if $f(W) = -
  W$, so the subset sum problem can also be solved by computing the
  objective function value at $b = W$.  
\end{proof}

\noindent
The uncertainty set~$U_{\hat c,\Gamma}$ is a polytope defined
explicitly by~$n+1$ linear inequalities. In particular, this shows
NP-hardness for polytopal uncertainty sets given by an outer
description:

\begin{corollary}
  The robust bilevel continuous knapsack problem with an uncertainty
  set being a polytope, given by a set of linear inequalities, is NP-hard.
\end{corollary}

Theorem~\ref{thm_simplex} also implies that the problem is NP-hard
under Gamma-uncertainty, which is sometimes also called budgeted
uncertainty~\cite{BertsimasSim}. However, this is only true in case
the total amount of deviation is bounded. Since uncertainty sets
cannot be replaced equivalently by their convex hulls in the bilevel
case, the above result does not necessarily hold for the original
definition of Gamma-uncertainty, where the number of deviating entries
is bounded.

It is easy to check that~$U_{\hat c,\Gamma}$ agrees with
$\text{conv}\{\hat{c}, \hat{c} + \Gamma e_1, \dots, \hat{c} + \Gamma
e_n\}$, where~$e_i$ denotes the $i$-th unit vector.  This proves

\begin{corollary}\label{cor_convhull}
  The robust bilevel continuous knapsack problem with an uncertainty
  set being the convex hull of a finite set of vectors, which are
  explicitly given as part of the input, is NP-hard.
\end{corollary}

It follows from Theorem~\ref{theorem_simplex_eval} that also the
evaluation of the leader's objective function is NP-hard in all cases
mentioned in the last corollaries.

Recall that the problem is tractable if the uncertainty set~$U$ is
finite and given explicitly as part of the input; see
Theorem~\ref{theorem_finite}. Together with
Corollary~\ref{cor_convhull}, this again shows that replacing the
uncertainty set by its convex hull may not only change the optimal
solution, but even the complexity of the problem significantly.

\begin{remark}
  In contrast to the result of Corollary~\ref{cor_convhull}, the
  problem can be solved in polynomial time if the uncertainty set is
  given as the convex hull of a constant number~$k$ of vectors. In
  fact, the number of possible sortings in the follower's solution can
  be bounded by $\mathcal{O}(n^{2 k})$ in this case, and they can be
  enumerated explicitly in polynomial time for constant~$k$. The
  tractability then follows from Theorem~\ref{theorem_finite}.
\end{remark}

\section{Norm uncertainty}\label{section_normuncertainty}

In Section~\ref{section_intervaluncertainty}, we have shown that the
robust bilevel continuous knapsack problem can be solved efficiently
if the uncertainty set is defined componentwise by intervals. This can
be seen as a special case of an uncertainty set defined by a $p$-norm,
$$U^p_{\hat c,\Gamma} := \{c \in \mathbb{R}^n \colon ||c-\hat c||_p\le \Gamma\}\;,$$
where $\hat{c} \in \mathbb{R}_{> 0}^n$, $\Gamma > 0$, and $p=\infty$.
On the other hand, the case~$p=1$ is closely related to the simplicial
case discussed in the previous section, which turned out to be
NP-hard.  In the following, we show NP-hardness for
all~$p\in[1,\infty)$. This is also of interest because the case $p =
2$ corresponds to an ellipsoidal uncertainty set, which is commonly
used in robust optimization.

\begin{theorem}\label{theorem_norm}
  Let $p\in[1,\infty)$. Then the robust bilevel continuous knapsack
  problem with uncertainty set $U^p_{\hat c,\Gamma}$, where $\hat{c}$
  and $\Gamma$ are part of the input, is NP-hard.  This remains true
  under the assumption $U^p_{\hat
    c,\Gamma}\subseteq\mathbb{R}^n_{>0}$.
\end{theorem}

\begin{proof}
  This can be shown very similarly to Theorem~\ref{thm_simplex}. We
  use the same instance, except that
  $$\hat{c} := (2Mw_1 - w_1^{1/p}, \dots, 2Mw_m - w_m^{1/p},
  2M^2)$$
  and $\Gamma:=W^{1/p}$. By the same reasoning as before, we may assume that an
  adversary's optimal solution satisfies $c_i \in \{\hat{c}_i,
  \hat{c}_i + w_i^{1 / p}\}$ for all~$i \in \{1, \dots, m\}$ and $c_{m
    + 1} = \hat{c}_{m + 1}$. For this, note that, although it is now
  allowed to deviate from $\hat{c}$ in any direction, the adversary
  has no incentive to choose $c_i < \hat{c}_i$ for any $i \in \{1,
  \dots, n\}$ in the given instance. In particular, it is cheaper in
  terms of the constraint $\sum_{i = 1}^n |c_i - \hat{c}_i|^p \leq
  \Gamma^p$ to shift all values $c_1/a_1, \dots, c_m/a_m$ to the right
  by some $\varepsilon > 0$ than moving $c_{m + 1}/a_{m + 1}$ to the
  left by $\varepsilon$, i.e.,
  $$\sum_{i = 1}^m (\varepsilon a_i)^p = \varepsilon^p \sum_{i = 1}^m w_i^p \leq \varepsilon^p (\sum_{i = 1}^m w_i)^p < \varepsilon^p M^p = (\varepsilon a_{m + 1})^p.$$
  The rest of the proof is analogous.  
\end{proof}

In the proof of Theorem~\ref{theorem_norm}, we implicitly assumed that
we can compute $p$-th roots in polynomial time when defining~$\hat c$
and~$\Gamma$. In general, these values cannot be computed or even
represented exactly in a polynomial time algorithm. However, they only
influence the set of possible sortings, not the actual leader's
solution and objective function values. For the former, a sufficiently
good approximation of the $p$-th roots, which can be computed in
polynomial time, leads to the same result.

\begin{remark}
  In the situation of Theorem~\ref{theorem_norm},
  Remark~\ref{rem_simplex} still applies: while the proof is
  formulated for the pessimistic setting, the variation proposed in
  Remark~\ref{rem_simplex} allows to show NP-hardness also in the
  optimistic setting.
\end{remark}

\noindent
In the case~$p=2$, the uncertainty set~$U^2_{\hat c,\Gamma}$ is an
axis-parallel ellipsoid. We thus derive the following result:
\begin{corollary}
  The robust bilevel continuous knapsack problem with (uncorrelated)
  ellipsoidal uncertainty is NP-hard.
\end{corollary}

\noindent
Finally, we can again show NP-hardness even for the problem of
evaluating the leader's objective function. The proof is analogous to
the case of simplicial uncertainty discussed in the previous section.

\begin{theorem}
  Let $p\in[1,\infty)$. Then evaluating the objective function of the
  robust bilevel continuous knapsack problem with uncertainty set
  $U^p_{\hat c,\Gamma}$, where $\hat{c}$ and $\Gamma$ are part of the
  input, is NP-hard.
\end{theorem}

\section{Uncertain item sizes} \label{section_uncertainitemsizes}

While our research was focused on uncertain follower's profits, one
may also consider the variant of the robust bilevel continuous
knapsack problem in which the item sizes~$a$ are uncertain. This was
investigated in~\cite{huegging}, considering different uncertainty
sets similarly to here. In case of discrete uncertainty, the algorithm
presented in Section~\ref{section_scenariouncertainty} could be
transferred directly. Regarding interval uncertainty, different types
of leader's objective functions were distinguished since the variants
(a), (b) and (c) mentioned in Section~\ref{section_withoutuncertainty}
are not equivalent anymore. Moreover, one has to be very careful about
the optimistic setting because with uncertain item sizes, the optimal
value might not be attained in the adversary's problem. Additionally,
one-point intersections of intervals in the optimistic setting are not
understood completely in this problem version yet; they can be
expected to be at least as complicated to handle as in
Remark~\ref{rem_interval}. However, many structural results were
established and an algorithm similar to the one presented in
Section~\ref{section_intervaluncertainty} was shown to solve the
problem under mild additional assumptions. Similarly to the results in
Section~\ref{section_discreteuncorrelateduncertainty} and
Section~\ref{section_simplicialuncertainty}, NP-hardness was also
proved for uncertain item sizes lying in a discrete uncorrelated or a
polytopal uncertainty set, assuming the pessimistic setting in the
latter case.

\section{Conclusion} \label{section_conclusion}

We have started the investigation of robust bilevel optimization by
addressing the bilevel continuous knapsack problem with uncertain
follower's objective. It turned out that standard results from
single-level robust optimization do not hold anymore: firstly, it is
not possible in general to replace an uncertainty set by its convex
hull without changing the problem. Secondly, the case of interval
uncertainty is not trivial anymore. Even if we have shown that it is
still tractable in the case of the bilevel continuous knapsack
problem, we conjecture that in general, the interval case cannot be
reduced to the certain problem anymore.

All hardness results presented in this paper are in the weak
sense. This leaves open the question whether the corresponding problem
variants are actually strongly NP-hard or whether pseudopolynomial
algorithms exist. This is left as future work, together with questions
of approximability.  In view of the remarks in
Section~\ref{section_uncertainitemsizes}, another task left for future
research is the simultaneous consideration of uncertain weights and
profits.  Another obvious research direction could be to consider
other bilevel optimization problems and to investigate how their
complexity increases when taking uncertainty into account.

\bibliographystyle{abbrv}
\bibliography{literature}

\end{document}